\documentclass{amsart}
\usepackage[latin9]{inputenc}
\usepackage{rotfloat}
\usepackage{amstext}
\usepackage{amsthm}
\usepackage{amssymb}
\usepackage{graphicx}

\makeatletter

\providecommand{\tabularnewline}{\\}

  \theoremstyle{definition}
  \newtheorem{defn}{\protect\definitionname}
\theoremstyle{plain}
\newtheorem{thm}{\protect\theoremname}
  \theoremstyle{plain}
  \newtheorem{lem}{\protect\lemmaname}

\usepackage[backend=bibtex,style=numeric,sorting=none,firstinits=true]{biblatex}
\usepackage[foot]{amsaddr}

\makeatother

  \providecommand{\definitionname}{Definition}
  \providecommand{\lemmaname}{Lemma}
\providecommand{\theoremname}{Theorem}

\begin{document}
\global\long\def\bias{\mathrm{bias}}
\global\long\def\var{\mathrm{var}}
\global\long\def\mse{\mathrm{MSE}}
\global\long\def\mise{\mathrm{MISE}}
\global\long\def\lognormalkernel{\mathrm{LN}}
\global\long\def\gammakernel{\mathrm{G}}
\global\long\def\bskernel{\mathrm{BS}}
\global\long\def\igkernel{\mathrm{IG}}
\global\long\def\rigkernel{\mathrm{RIG}}
\global\long\def\euler{\mathrm{e}}

\title[Unified treatment of the asymptotics of asymmetric KDEs]{Unified treatment of the asymptotics of asymmetric kernel density estimators}

\author{Till Hoffmann}

\address{Department of Mathematics, Imperial College London, SW7 2AZ, United Kingdom}

\email{t.hoffmann13@imperial.ac.uk}

\author{Nick Jones}

\email{n.jones@imperial.ac.uk}
\begin{abstract}
We extend balloon and sample-smoothing estimators, two types of variable-bandwidth kernel density estimators, by a shift parameter and derive their asymptotic properties. Our approach facilitates the unified study of a wide range of density estimators which are subsumed under these two general classes of kernel density estimators. We demonstrate our method by deriving the asymptotic bias, variance, and mean (integrated) squared error of density estimators with gamma, log-normal, Birnbaum-Saunders, inverse Gaussian and reciprocal inverse Gaussian kernels. We propose two new density estimators for positive random variables that yield properly-normalised density estimates. Plugin expressions for bandwidth estimation are provided to facilitate easy exploratory data analysis.
\end{abstract}

\maketitle

\section{Introduction}

Kernel density estimation with fixed bandwidth has become the most common method to estimate unknown densities of continuous random variables. The standard kernel density estimator (KDE) is defined as 
\begin{equation}
\hat{f}\left(x\right)=\frac{1}{nh}\sum_{i=1}^{n}K\left(\frac{X_{i}-x}{h}\right),\label{eq:standard-kde}
\end{equation}
where $\left\{ X_{i}:i=1,\ldots,n\right\} $ is a set of $n$ i.i.d. samples of a univariate random variable $X$ drawn from a generating distribution $f$, $K$ is a non-negative function that integrates to unity, and $h$ is the bandwidth which controls the smoothness of the estimate \cite{Silverman1986}. We assume that $f$ is supported on the real line. Unless otherwise stated $\int$ will refer to an integral over the interval $\left(-\infty,\infty\right)$. 

Before outlining the remainder of this paper, we briefly review three topics for which we will provide a unified treatment: variable-bandwidth KDEs, shifted KDEs and KDEs with asymmetric kernels.

\emph{Variable-bandwidth KDEs: }Choosing the fixed bandwidth $h$ entails a trade-off. If the bandwidth is small, high-density regions are smoothed appropriately and we can recover small-scale features. Yet low-density regions will be undersmoothed. If the bandwidth is large, low-density regions are smoothed appropriately. Unfortunately, high-density regions will be oversmoothed and small-scale features will be masked. Terrell and Scott considered two variable-bandwidth KDEs to alleviate this problem \cite{Terrell1992}. \emph{Balloon estimators} have a bandwidth that depends on the \emph{evaluation point}, i.e. the point, $x$, at which $f\left(x\right)$ needs to be estimated,
\begin{equation}
\hat{f}\left(x\right)=\frac{1}{nh\left(x\right)}\sum_{i=1}^{n}K\left(\frac{X_{i}-x}{h\left(x\right)}\right).\label{eq:balloon-estimator}
\end{equation}
By contrast, the bandwidth of \emph{sample-smoothing estimators} depends on the sample associated with each kernel such that
\begin{equation}
\hat{f}\left(x\right)=\frac{1}{n}\sum_{i=1}^{n}\frac{1}{h\left(X_{i}\right)}K\left(\frac{X_{i}-x}{h\left(X_{i}\right)}\right).\label{eq:smoothing-estimator}
\end{equation}
The latter yield properly normalised density estimates, i.e. $\int dx\,\hat{f}\left(x\right)=1$, because each term of the sum in Eq.~\eqref{eq:smoothing-estimator} integrates to unity. In general, balloon estimators are not proper densities themselves.

\emph{Shifted KDEs: }Standard KDEs tend to overestimate densities in regions with few samples and underestimate densities in regions with many samples. Samiuddin and El-Sayyad proposed to shift samples slightly from low-density regions to high-density regions to achieve lower bias \cite{Samuiddin1990}. Hall and Minnotte extended this approach to decrease the bias further using data sharpening estimators \cite{Hall2002a}. KDEs with variable bandwidth and shifted samples were considered by Jones et al. in the context of sample-smoothing estimators \cite{Jones1994}. 

\emph{Asymmetric KDEs:} A literature on using asymmetric kernels for density estimation has developed in parallel: gamma \cite{Chen2000}, log-normal and Birnbaum-Saunders \cite{Jin2003}, as well as inverse Gaussian and reciprocal inverse Gaussian \cite{Scaillet2004} eliminate bias on or near the boundary of a (semi-)bounded interval. Whenever a new kernel is being considered, its asymptotic properties have to be carefully derived.

In the following, we provide a more general treatment of shifted sample-smoothing estimators and introduce shifted balloon estimators, and develop a generic framework to obtain the asymptotic properties of a range of kernel density estimators. After discussing mathematical preliminaries in Section~\ref{sec:preliminaries}, we define shifted balloon and sample-smoothing estimators in Sections~\ref{sec:balloon} and \ref{sec:smoothing}, respectively. We derive the asymptotic bias, variance and mean (integrated) squared error for each class of estimators. In Section~\ref{sec:application}, we apply our results to density estimators with gamma, log-normal, Birnbaum-Saunders, inverse Gaussian and reciprocal inverse Gaussian kernels, derive plugin expressions for bandwidth selection, and propose two new density estimators for density estimation on the positive real line. We conclude in Section~\ref{sec:discussion}.

\section{Preliminaries\label{sec:preliminaries}}

The estimators in Eqs.~\eqref{eq:standard-kde} to \eqref{eq:smoothing-estimator} can all be expressed as general weight function estimators defined by
\begin{align}
\hat{f}\left(x\right) & =\frac{1}{n}\sum_{i=1}^{\infty}W\left(X_{i},x\right),\label{eq:generalised-estimator}
\end{align}
where $W\left(y,x\right)$ is a non-negative function that peaks near $x=y$. Because the samples are i.i.d., the first moment of the estimator is
\begin{align}
\left\langle \hat{f}\left(x\right)\right\rangle  & =\left\langle W\left(X,x\right)\right\rangle ,\label{eq:generalised-mean}
\end{align}
where the expectation is with respect to the samples $X$. The second moment is
\begin{align*}
\left\langle \hat{f}^{2}\left(x\right)\right\rangle  & =\frac{1}{n^{2}}\left\langle \sum_{i,j=1}^{n}W\left(X_{i},x\right)W\left(X_{j},x\right)\right\rangle \\
 & =\frac{1}{n}\left\langle W^{2}\left(X,x\right)\right\rangle +\frac{n-1}{n}\left\langle W\left(X,x\right)\right\rangle ^{2}
\end{align*}
such that the variance becomes
\begin{align}
\var\hat{f}\left(x\right) & =\left\langle \hat{f}^{2}(x)\right\rangle -\left\langle \hat{f}(x)\right\rangle ^{2}\nonumber \\
 & =\frac{1}{n}\left(\left\langle W^{2}\left(X,x\right)\right\rangle -\left\langle W\left(X,x\right)\right\rangle ^{2}\right).\label{eq:generalised-variance}
\end{align}

The canonical criterion for assessing the local quality of density estimates is the mean squared error (MSE) and is defined as

\[
\mse\hat{f}\left(x\right)=\left\langle \left(f(x)-\hat{f}(x)\right)^{2}\right\rangle 
\]
for a given evaluation point $x$. We rewrite the MSE in terms of squared bias and variance such that 
\begin{align}
\mse\hat{f}(x) & =\bias^{2}\hat{f}(x)+\var\hat{f}(x),\nonumber \\
\text{where}\quad\bias\hat{f}(x) & =\left\langle \hat{f}(x)\right\rangle -f(x)\label{eq:generalised-bias}\\
 & =\left\langle W\left(X,x\right)\right\rangle -f\left(x\right).\nonumber 
\end{align}
The mean integrated squared error (MISE) is defined as
\[
\mise\hat{f}=\int dx\,\mse\hat{f}\left(x\right)
\]
and will be our criterion for the global quality of density estimates.

The balloon estimator in Eq.~\eqref{eq:balloon-estimator}, the sample-smoothing estimator in Eq.~\eqref{eq:smoothing-estimator}, and the estimators we consider in the following two sections can be expressed in terms of a kernel. Following Terrell and Scott, we formalise the notion of a kernel in the following definition \cite{Terrell1992}.
\begin{defn}
\label{def:order-p-kernel}An order-$p$ kernel is a univariate, non-negative function $K$ such that 
\begin{equation}
\begin{aligned}m_{k} & \equiv\int dz\,K(z)z^{k} & \kappa & \equiv\int dz\,K^{2}(z)\\
m_{k} & =\begin{cases}
1 & \text{if }k=0,\,p\\
0 & \text{if }0<k<p
\end{cases} & 0 & <\kappa<\infty,
\end{aligned}
\label{eq:kernel-properties}
\end{equation}
where $m_{k}$ is the $k^{\mathrm{th}}$ moment of the kernel.
\end{defn}

\section{Shifted balloon estimator\label{sec:balloon}}

We now provide our modified version of the balloon estimator.
\begin{defn}
A shifted balloon estimator is a kernel density estimator whose bandwidth depends on the point $x$ at which it is evaluated, i.e. 
\begin{equation}
\hat{f}(x)=\frac{1}{n}\sum_{i=1}^{n}\frac{1}{h(x)}K\left(\frac{X_{i}-x-h^{p}\left(x\right)\delta(x)}{h(x)}\right),\label{eq:balloon-definition}
\end{equation}
where $p$ is the order of the kernel $K$, and $\delta(x)$ is a shift which remains bounded as $n\rightarrow\infty$, i.e. $\lim_{n\rightarrow\infty}\left|\delta\left(x\right)\right|<\infty$.
\end{defn}
The kernel is thus shifted with respect to the sample by a small amount $h^{p}\left(x\right)\delta\left(x\right)$ which depends on the evaluation point and vanishes as the bandwidth decreases. Shifted balloon estimators do not yield properly-normalised density estimates in general, similar to their unshifted counterparts.

We need to consider the bias and variance to evaluate the MSE in the limit of small bandwidth.
\begin{thm}
\label{thm:balloon-moments} The mean and variance of a shifted balloon estimator with order-$p$ kernel are 
\begin{align}
\left\langle \hat{f}\left(x\right)\right\rangle  & =A_{0}\left(x\right)+A_{p}\left(x\right)+o\left(h^{p}\left(x\right)\right)\label{eq:balloon-mean}\\
\text{and}\quad\var\hat{f}_{h}\left(x\right) & =\frac{f\left(x\right)\kappa}{nh\left(x\right)}+o\left(n^{-1}h^{-1}\left(x\right)\right),\label{eq:balloon-var}\\
\text{where}\quad A_{k}\left(x\right) & =h^{k}\left(x\right)\sum_{j=k}^{\infty}\frac{f^{\left(j\right)}\left(x\right)}{k!\left(j-k\right)!}\left[h^{p}\left(x\right)\delta\left(x\right)\right]^{j-k}\label{eq:balloon-coefficient}
\end{align}
and $f^{(j)}(x)$ denotes the $j^{\text{th}}$ derivative of $f$ evaluated at $x$. The factorial is defined such that $0!\equiv1$.
\end{thm}
A proof is given in Appendix \ref{app:proofs}. Substituting Eqs.~\eqref{eq:balloon-coefficient} and \eqref{eq:balloon-mean} into Eq.~\eqref{eq:generalised-bias} and dropping higher-order terms in $h\left(x\right)$, the bias is 
\begin{equation}
\bias\hat{f}\left(x\right)=h^{p}\left(x\right)\delta\left(x\right)f'\left(x\right)+\frac{h^{p}\left(x\right)}{p!}f^{\left(p\right)}\left(x\right)+o\left(h^{p}\left(x\right)\right),\label{eq:balloon-bias}
\end{equation}
which is the same as the expression derived by Terrell and Scott if we let $\delta\left(x\right)=0$ \cite{Terrell1992}. We thus see from Eq.~\eqref{eq:balloon-bias} that any dependence of the bias on the first derivative of the generating distribution is due to the shift between the kernel mean and the associated sample as Chen noted in the context of density estimators using gamma kernels \cite{Chen2000}. Adding Eq.~\eqref{eq:balloon-var} and the square of Eq.~\eqref{eq:balloon-bias} yields the MSE
\begin{multline*}
\mse\hat{f}\left(x\right)=\left[h^{p}\left(x\right)\delta\left(x\right)f'\left(x\right)+\frac{h^{p}\left(x\right)}{p!}f^{\left(p\right)}\left(x\right)\right]^{2}+\frac{f\left(x\right)\kappa}{nh\left(x\right)}\\
+o\left(n^{-1}h^{-1}\left(x\right)+h^{2p}\left(x\right)\right).
\end{multline*}

\section{Shifted sample-smoothing estimators\label{sec:smoothing}}

In the following, we discuss our modified version of the sample-smoothing estimator.
\begin{defn}
A shifted sample-smoothing estimator is a kernel density estimator whose bandwidth depends on the samples $X_{i}$, i.e. 
\begin{equation}
\hat{f}\left(x\right)=\frac{1}{n}\sum_{i=1}^{n}\frac{1}{h\left(X_{i}\right)}K\left(\frac{X_{i}-x+h^{p}\left(X_{i}\right)\delta\left(X_{i}\right)}{h\left(X_{i}\right)}\right),\label{eq:smoothing-definition}
\end{equation}
where $p$ is the order of the kernel $K$, and $\delta(X_{i})$ is a shift which remains bounded as $n\rightarrow\infty$, i.e. $\lim_{n\rightarrow\infty}\left|\delta\left(X_{i}\right)\right|<\infty$.
\end{defn}
Again, the kernel is shifted with respect to the sample by a small amount but the magnitude of the shift depends on the sample and not the evaluation point. Shifted sample-smoothing estimators are guaranteed to yield properly-normalised density estimates.

As in the previous section, we consider the mean and the variance assuming the bandwidth is small. See Appendix \ref{app:proofs} for a proof of the following theorem.
\begin{thm}
\label{thm:smoothing-moments} If 
\begin{equation}
z=\frac{y-x+h^{p}\left(y\right)\delta\left(y\right)}{h\left(y\right)}\label{eq:change-of-variables-condition}
\end{equation}
is a monotonic function of $y$ for all $x$, the mean and variance of a shifted sample-smoothing estimator with order-$p$ kernel are 
\begin{align}
\left\langle \hat{f}\left(x\right)\right\rangle  & =B_{0}\left(x\right)+B_{p}\left(x\right)+o\left(h^{p}\left(x\right)\right)\nonumber \\
\text{and}\quad\var\hat{f}\left(x\right) & =\frac{f(x)\kappa}{nh(x)}+o\left(n^{-1}h^{-1}\left(x\right)\right),\label{eq:smoothing-var}\\
\text{where}\quad B_{k} & =\sum_{j=k}^{\infty}\frac{1}{k!\left(j-k\right)!}\frac{d^{j}}{dx^{j}}\left(f\left(x\right)h^{k+p\left(j-k\right)}\left(x\right)\left(-\delta\left(x\right)\right)^{j-k}\right).\nonumber 
\end{align}

\end{thm}
The condition in Eq.~\eqref{eq:change-of-variables-condition} is satisfied by most kernels although exceptions exist as we will see in the next section. Using the same argument as in the previous section and dropping higher-order terms in $h\left(x\right)$, we have

\begin{align}
\bias\hat{f}\left(x\right) & =-\frac{d}{dx}\left(f\left(x\right)h^{p}\left(x\right)\delta\left(x\right)\right)+\frac{1}{p!}\frac{d^{p}}{dx^{p}}\left(h^{p}\left(x\right)f\left(x\right)\right)+o\left(h^{p}\left(x\right)\right).\label{eq:smoothing-bias}
\end{align}
The expressions agree with Terrell and Scott if we let $\delta\left(x\right)=0$ \cite{Terrell1992}. Jones et al. provide an alternative derivation for second-order kernels, i.e. $p=2$ \cite{Jones1994}. The MSE is 
\begin{multline*}
\mse\hat{f}(x)=\left[-\frac{d}{dx}\left(f\left(x\right)h^{p}\left(x\right)\delta\left(x\right)\right)+\frac{1}{p!}\frac{d^{p}}{dx^{p}}\left(h^{p}\left(x\right)f\left(x\right)\right)\right]^{2}+\frac{\kappa f\left(x\right)}{nh\left(x\right)}\\
+o\left(n^{-1}h^{-1}\left(x\right)+h^{2p}\left(x\right)\right).
\end{multline*}

\section{Application\label{sec:application}}

Having derived the asymptotic properties of shifted balloon and shifted sample-smoothing estimators allows us to obtain the asymptotic properties of a range of kernel density estimators easily. Before applying our method, we need to make a distinction that is immaterial for symmetric kernels. Recall the definition of the general weight function estimator which we repeat here for ease of reference: 
\begin{equation}
\hat{f}\left(x\right)=\frac{1}{n}\sum_{i=1}^{n}W\left(X_{i},x\right).\label{eq:generalised-estimator-2}
\end{equation}
There are two fundamentally different choices for the weight function.
\begin{defn}
A weight function $W\left(y,x\right)$ is called \emph{proper} if and only if 
\[
\int dy\,W\left(y,x\right)=1.
\]
It is called \emph{improper} otherwise.
\end{defn}
Estimators with proper weight functions are guaranteed to be densities themselves because each term in Eq.~\eqref{eq:generalised-estimator-2} integrates to unity. Sample-smoothing estimators belong to this class. Estimators with improper weight functions, such as balloon estimators, do not in general yield normalised densities. If the weight function is symmetric with respect to exchange of its argument, the distinction becomes immaterial. Standard kernels with fixed bandwidth as defined in Eq.~\eqref{eq:standard-kde} are examples of symmetric weight functions if the kernel is an even function of its argument.

\subsection{Density estimation using improper weight functions}

Improper weight functions have recently become popular for estimating densities with (semi-)bounded support to avoid boundary bias. See Section 2.11 in \cite{Wand1995} for a detailed discussion of boundary bias. In particular, Brown and Chen used beta-distribution kernels for non-parametric regression on a finite interval \cite{Brown1999}. Chen proposed the use of gamma-distribution kernels to estimate densities of positive random variables \cite{Chen2000}, which spurred further research into positive density estimation using Birnbaum-Saunders and log-normal kernels \cite{Jin2003} as well as inverse Gaussian and reciprocal inverse Gaussian kernels \cite{Scaillet2004}. 

Intuitively, the contribution from any weight function becomes ever more sharply peaked as the bandwidth of the estimator decreases. Because the global properties become less important, we should be able to approximate general weight functions by simpler kernels that peak where the sample and evaluation point are close. We will see that the improper estimators in the literature are asymptotically equivalent to shifted balloon estimators as defined in Eq.~\eqref{eq:balloon-definition}. This observation enables us to derive their asymptotic properties using the general results from Section \ref{sec:balloon} instead of having to derive them independently. We consider the gamma-kernel density estimator proposed by Chen \cite{Chen2000} in detail and discuss other kernels in Section~\ref{subsec:other-kernels}.

After reparametrisation, Chen's improper estimator is given by
\[
\hat{f}\left(x\right)=\sum_{i=1}^{n}\mathrm{G}\left(X_{i};1+\frac{x}{\sigma^{2}},\sigma^{2}\right),
\]
where $\sigma$ is a bandwidth parameter, and $\mathrm{G}\left(t;k,\theta\right)=\frac{\theta^{-k}}{\Gamma\left(k\right)}x^{k-1}\exp\left(-\frac{t}{\theta}\right)$ denotes the probability density function of a gamma random variable $t$ with shape parameter $k$ and scale parameter $\theta$. If the bandwidth $\sigma$ is small, the shape parameter of the gamma distribution is large and it can be approximated by a Gaussian with the same mean \cite{Johnson1994}, i.e.
\[
\mathrm{G}\left(X_{i};1+\frac{x}{\sigma^{2}},\sigma^{2}\right)\rightarrow\mathcal{N}\left(X_{i};x+\sigma^{2},\sigma^{2}\left(x+\sigma^{2}\right)\right).
\]
The approximation takes the form of a shifted balloon estimator with Gaussian kernel and 
\begin{align*}
h\left(x\right) & =\sigma\sqrt{x+\sigma^{2}}\\
h^{2}\left(x\right)\delta\left(x\right) & =\sigma^{2}.
\end{align*}

Having represented Chen's improper estimator as a shifted balloon estimator, we can find the bias and variance using Eqs.~\eqref{eq:balloon-bias} and \eqref{eq:balloon-var}:
\begin{align*}
\bias\hat{f}\left(x\right) & =\sigma^{2}f'\left(x\right)+\frac{\sigma^{2}\left(x+\sigma^{2}\right)}{2}f''\left(x\right)+o\left(\sigma^{2}\right)\\
 & =\sigma^{2}\left(f'\left(x\right)+\frac{x}{2}f''\left(x\right)\right)+o\left(\sigma^{2}\right)\\
\text{and}\quad\var\hat{f}\left(x\right) & =\frac{f\left(x\right)}{2n\sigma\sqrt{\pi\left(x+\sigma^{2}\right)}}+o\left(n^{-1}\sigma^{-1}\right)\\
 & =\frac{f\left(x\right)}{2n\sigma\sqrt{\pi x}}+o\left(n^{-1}\sigma^{-1}\right),
\end{align*}
where we have used that Gaussians are second-order kernels with $\kappa=\frac{1}{2\sqrt{\pi}}$. The expressions agree with the ones derived by Chen. The MSE is
\[
\mse\hat{f}\left(x\right)=\sigma^{4}\left(f'\left(x\right)+\frac{x}{2}f''\left(x\right)\right)^{2}+\frac{f\left(x\right)}{2n\sigma\sqrt{\pi x}}+o\left(\sigma^{4}+n^{-1}\sigma^{-1}\right).
\]

\begin{figure}
\begin{centering}
\includegraphics{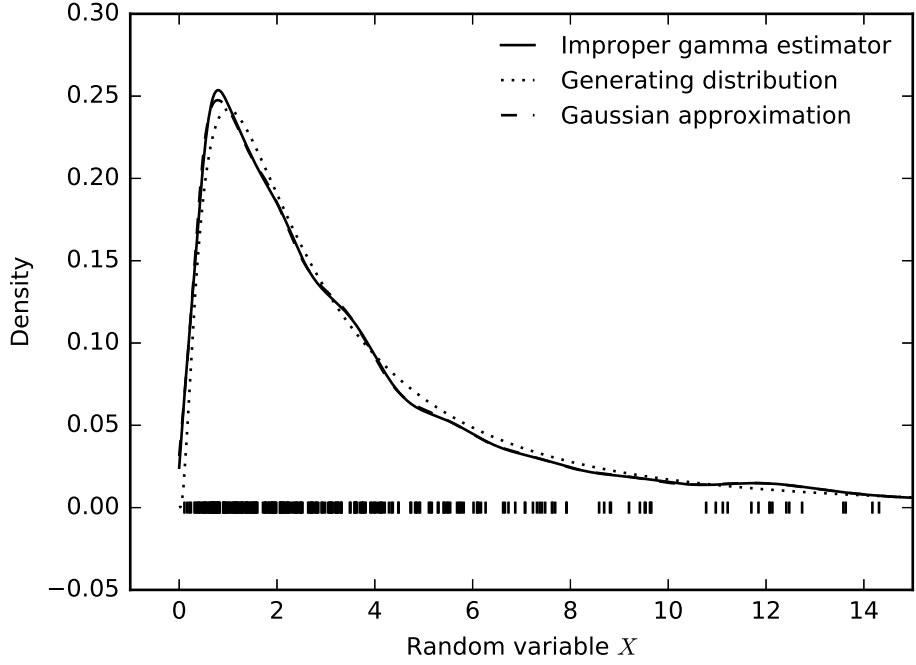}
\par\end{centering}

\caption{\label{fig:improper-gamma}Density estimate of $n=300$ samples drawn from a log-normal distribution with logarithmic mean $\mu=1$ and logarithmic variance $\Sigma^{2}=1$ using an improper gamma estimator (solid black line) and a Gaussian approximation (dashed black line). The generating distribution is shown as a dotted black line and individual samples are shown as markers on the horizontal axis. The bandwidth for both estimators was obtained using the plugin expression in Eq.~\eqref{eq:chen-plugin} which minimises the MISE in Eq.~\eqref{eq:chen-lognormal-mise} as shown in Figure~\ref{fig:bandwidth-comparison}.}
\end{figure}

Extending the work by Chen \cite{Chen2000}, we approximate $f\left(x\right)$ by a reference distribution with known functional form which enables us to obtain a plugin expression for selecting the bandwidth. The parametric form of the reference distribution is largely a free choice; we assume that $f(x)$ is approximately log-normal with logarithmic mean $\mu$ and variance $\Sigma^{2}$ for convenience. Integrating yields 
\begin{equation}
\mise\hat{f}=\sigma^{4}\frac{\exp\left(\frac{9\Sigma^{2}}{4}-3\mu\right)\left(12+20\Sigma^{2}+9\Sigma^{4}\right)}{128\sqrt{\pi}\Sigma^{4}}+\frac{\exp\left(\frac{\Sigma^{2}}{8}-\frac{\mu}{2}\right)}{2n\sigma\sqrt{\pi}}\label{eq:chen-lognormal-mise}
\end{equation}
and setting the derivative of the MISE with respect to $\sigma$ to zero gives the optimal scale parameter 
\begin{equation}
\sigma^{*}=\frac{2^{4/5}\Sigma\exp\left(\frac{\mu}{2}-\frac{17\Sigma^{2}}{40}\right)}{\left(12+20\Sigma^{2}+9\Sigma^{4}\right)^{1/5}}n^{-1/5}.\label{eq:chen-plugin}
\end{equation}
The expression allows us to obtain a bandwidth by estimating $\mu$ and $\Sigma$ from the samples without the need for computationally-expensive cross-validation. Plugin bandwidth estimators are particularly useful for preliminary data exploration. An example density estimate is shown in Figure~\ref{fig:improper-gamma}.

\subsection{Density estimation using proper weight functions}

Improper weight functions have the unappealing property that they do not yield properly normalised density estimates in general. Jeon and Kim proposed a density estimator using a proper gamma kernel to alleviate this problem \cite{Jeon2014}. After reparametrisation, their estimator is given by
\begin{equation}
\hat{f}\left(x\right)=\frac{1}{n}\sum_{i=1}^{\infty}\mathrm{G}\left(x;1+\frac{X_{i}}{\sigma^{2}},\sigma^{2}\right),\label{eq:jeon-kim-estimator}
\end{equation}
i.e. the role of the evaluation point and the samples have been swapped. The kernel can be approximated by a normal distribution if the bandwidth is small and the estimator takes the form of a shifted sample-smoothing estimator
\[
\mathrm{G}\left(x;1+\frac{X_{i}}{\sigma^{2}},\sigma^{2}\right)\rightarrow\mathcal{N}\left(x;X_{i}+\sigma^{2},\sigma^{2}\left(X_{i}+\sigma^{2}\right)\right).
\]
Using Eqs.~\eqref{eq:smoothing-bias} and \eqref{eq:smoothing-var}, the bias and variance are
\begin{align*}
\bias\hat{f}\left(x\right) & =-\sigma^{2}f'\left(x\right)+\frac{1}{2}\frac{d^{2}}{dx^{2}}\left(\sigma^{2}\left(x+\sigma^{2}\right)f\left(x\right)\right)+o\left(\sigma^{2}\right)\\
 & =\frac{\sigma^{2}x}{2}f''\left(x\right)+o\left(\sigma^{2}\right)\\
\text{and}\quad\var\hat{f}\left(x\right) & =\frac{f\left(x\right)}{2n\sigma\sqrt{\pi x}}+o\left(n^{-1}\sigma^{-1}\right).
\end{align*}
Jeon and Kim obtained an upper bound for the MSE of the estimator defined in Eq.~\eqref{eq:jeon-kim-estimator} but we can provide an explicit form: 
\[
\mse\hat{f}\left(x\right)=\frac{\sigma^{4}}{4}\left[xf''(x)\right]^{2}+\frac{f(x)}{2n\sigma\sqrt{\pi x}}+o\left(\sigma^{4}+n^{-1}\sigma^{-1}\right).
\]
Employing the log-normal distribution as a reference distribution, and optimising the MISE with respect to $\sigma$ yields the plugin expression 
\[
\sigma^{*}=\frac{2^{4/5}\Sigma\exp\left(\frac{\mu}{2}-\frac{17\Sigma^{2}}{40}\right)}{\left(12+4\Sigma^{2}+\Sigma^{4}\right)^{1/5}}n^{-1/5}.
\]

\begin{figure}
\begin{centering}
\includegraphics{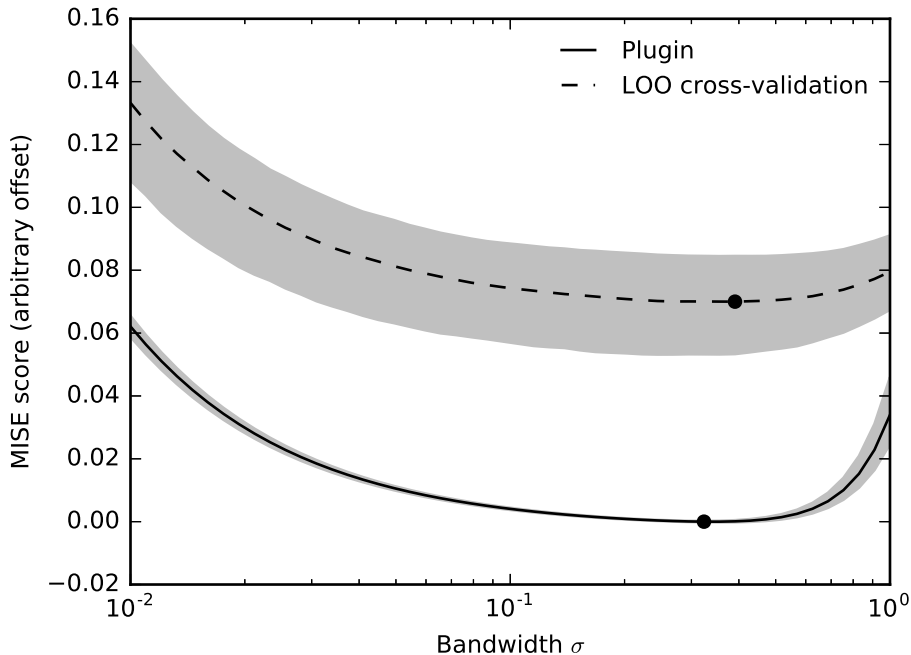}
\par\end{centering}

\caption{\label{fig:bandwidth-comparison}MISE profiles as a function of bandwidth $\sigma$ using a Gaussian approximation (solid line) and leave-one-out (LOO) cross-validation (dashed line). The shaded bands correspond to the 90\% confidence interval from one thousand independent samples of size $n=300$ drawn from a log-normal distribution with $\mu=\Sigma=1$. The offset between the two curves is arbitrary.}
\end{figure}
As discussed in Appendix~\ref{app:cross-validation}, cross-validation is straightforward for Jeon and Kim's estimator. A comparison between the plugin expression and leave-one-out cross-validation is shown in Figure~\ref{fig:bandwidth-comparison}. The results are encouraging even for a modest sample size of $n=300$. The MISE profile obtained from cross-validation is more variable because it depends on each individual sample whereas the asymptotic MISE profile only depends on the logarithmic mean and variance of the whole sample.

\subsection{Application to estimators with other kernels\label{subsec:other-kernels}}

\begin{sidewaystable}
\vspace{.65\textwidth}%
\begin{tabular*}{0.9\textheight}{@{\extracolsep{\fill}}llllll}
Distribution & $\left\langle t\right\rangle $ & Variance & Parametrisation & $h^{2}\left(x\right)\delta\left(x\right)$ & $h\left(x\right)$\tabularnewline
\hline 
$\gammakernel\left(t;k,\theta\right)=\frac{\theta^{-k}t^{k-1}}{\Gamma(k)}\exp\left(-\frac{t}{\theta}\right)$ & $k\theta$ & $k\theta^{2}$ & $\gammakernel\left(y;\frac{x}{\sigma^{2}}+1,\sigma^{2}\right)$ & $\sigma^{2}$ & $\sigma\sqrt{x+\sigma^{2}}$\tabularnewline
$\lognormalkernel\left(t;\mu,\sigma^{2}\right)=\frac{1}{\sqrt{2\pi}\sigma x}\exp\left(-\frac{\left(\log t-\mu\right)^{2}}{2\sigma^{2}}\right)$ & $\euler^{\mu+\frac{\sigma^{2}}{2}}$ & $\left\langle t\right\rangle ^{2}\left(\euler^{\sigma^{2}}-1\right)$ & $\lognormalkernel\left(y;\log x,\sigma^{2}\right)$ & $x\left(\euler^{\frac{\sigma^{2}}{2}}-1\right)$ & $x\euler^{\frac{\sigma^{2}}{2}}\sqrt{\euler^{\sigma^{2}}-1}$\tabularnewline
$\bskernel\left(t;\sigma,\lambda\right)=\frac{(1+t\lambda)}{2\sigma\sqrt{2\pi\lambda t^{3}}}\exp\left(-\frac{(t\lambda-1)^{2}}{2t\lambda\sigma^{2}}\right)$ & $\frac{2+\sigma^{2}}{2\lambda}$ & $\frac{\sigma^{2}\left(4+5\sigma^{2}\right)}{4\lambda^{2}}$ & $\bskernel\left(y;\sigma,x^{-1}\right)$ & $\frac{x\sigma^{2}}{2}$ & $x\sigma\sqrt{1+5\sigma^{2}/4}$\tabularnewline
$\igkernel\left(t;\mu,\lambda\right)=\sqrt{\frac{\lambda}{2\pi t^{3}}}\exp\left(-\frac{\lambda(t-\mu)^{2}}{2t\mu^{2}}\right)$ & $\mu$ & $\frac{\mu^{3}}{\lambda}$ & $\igkernel\left(y;x,\sigma^{-2}\right)$ & 0 & $\sigma x^{3/2}$\tabularnewline
$\rigkernel\left(t;\mu,\lambda\right)=\sqrt{\frac{\lambda}{2\pi t}}\exp\left(-\frac{\lambda(\mu t-1)^{2}}{2t\mu^{2}}\right)$ & $\frac{1}{\mu}+\frac{1}{\lambda}$ & $\frac{1}{\lambda\mu}+\frac{2}{\lambda^{2}}$ & $\rigkernel\left(y;\frac{1}{x-\sigma^{2}},\sigma^{-2}\right)$ & 0 & $\sigma\sqrt{x+\sigma^{2}}$\tabularnewline
\end{tabular*}

\vspace{2em}
\caption{\label{tbl:kernels}List of probability distributions with their mean and variance. The fourth column shows the parametrisation used for improper weight functions, where $x$ is the evaluation point and $y$ is the associated sample. The parametrisation for proper weight functions can be obtained by swapping $x$ and $y$. The fifth column lists the shift between the mean of the kernel and the sample. The effective bandwidth is shown in the last column. }

\vspace{2em}
\begin{tabular*}{0.9\textheight}{@{\extracolsep{\fill}}llllll}
 &  & \multicolumn{2}{l}{Improper weight function} & \multicolumn{2}{l}{Proper weight function}\tabularnewline
Kernel with ref. & $n\times\var_{\sigma}$  & $\bias_{\sigma}$  & $n^{1/5}\times\sigma*$ & $\bias_{\sigma}$  & $n^{1/5}\times\sigma*$\tabularnewline
\hline 
G \cite{Chen2000,Jeon2014} & $\frac{f(x)}{2\sigma\sqrt{\pi x}}$  & $\sigma^{2}\left(f'(x)+\frac{xf''(x)}{2}\right)$ & $\frac{2^{4/5}\Sigma\exp\left(\frac{\mu}{2}-\frac{17\Sigma^{2}}{40}\right)}{\left(12+20\Sigma^{2}+9\Sigma^{4}\right)^{1/5}}$ & $\frac{\sigma^{2}}{2}xf''\left(x\right)$ & $\frac{2^{4/5}\Sigma\exp\left(\frac{\mu}{2}-\frac{17\Sigma^{2}}{40}\right)}{\left(12+4\Sigma^{2}+\Sigma^{4}\right)^{1/5}}$\tabularnewline
LN \cite{Jin2003} & $\frac{f(x)}{2\sqrt{\pi}\sigma x}$  & $\frac{\sigma^{2}x}{2}\left(f'\left(x)+xf''\left(x\right)\right)\right)$ & $\frac{2^{4/5}\Sigma\exp\frac{\Sigma^{2}}{20}}{\left(12+4\Sigma^{2}+\Sigma^{4}\right)^{1/5}}$ & $\frac{\sigma^{2}}{2}\frac{d}{dx}\left(x\frac{d}{dx}\left(xf\left(x\right)\right)\right)$ & $\frac{2^{4/5}\Sigma\exp\frac{\Sigma^{2}}{20}}{\left(12+4\Sigma^{2}+\Sigma^{4}\right)^{1/5}}$\tabularnewline
BS \cite{Jin2003} & \multicolumn{5}{l}{same as log-normal}\tabularnewline
IG \cite{Scaillet2004} & $\frac{f(x)}{2\sqrt{\pi}\sigma x^{3/2}}$  & $\frac{\sigma^{2}}{2}x^{3}f''\left(x\right)$ & $\frac{2^{4/5}\Sigma\exp\left(\frac{7\Sigma^{2}}{40}-\frac{\mu}{2}\right)}{\left(12+68\Sigma^{2}+225\Sigma^{4}\right)^{1/5}}$ & \multicolumn{2}{c}{method not applicable}\tabularnewline
RIG \cite{Scaillet2004} & $\frac{f(x)}{2\sigma\sqrt{\pi x}}$  & $\frac{\sigma^{2}}{2}xf''\left(x\right)$ & $\frac{2^{4/5}\Sigma\exp\left(\frac{\mu}{2}-\frac{17\Sigma^{2}}{40}\right)}{\left(12+4\Sigma^{2}+\Sigma^{4}\right)^{1/5}}$ & $\sigma^{2}\left(f'(x)+\frac{xf''(x)}{2}\right)$ & $\frac{2^{4/5}\Sigma\exp\left(\frac{\mu}{2}-\frac{17\Sigma^{2}}{40}\right)}{\left(12+20\Sigma^{2}+9\Sigma^{4}\right)^{1/5}}$\tabularnewline
\end{tabular*}

\vspace{2em}
\caption{\label{tbl:asymptotic-properties}List of kernels with their asymptotic variance. The bias and optimal plugin bandwidth $\sigma^{*}$ are shown separately for proper and improper estimators. }
\end{sidewaystable}
We use the methods developed above to derive the asymptotic properties of estimators with proper and improper weight functions. The gamma, log-normal (LN), Birnbaum-Saunders (BS), inverse Gaussian (IG) and reciprocal inverse Gaussian (RIG) distributions are listed in Table~\ref{tbl:kernels}. To serve as a kernel, we parametrise the distributions in terms of the evaluation point $x$, associated sample $y$, and bandwidth parameter $\sigma$. In the limit $\sigma\rightarrow0$, all kernels converge to a normal distribution \cite{Johnson1994} with mean $x+h^{2}\left(x\right)\delta\left(x\right)$ and variance $h^{2}\left(x\right)$ such that our method is applicable. The parametrisation for proper estimators can be obtained by exchanging the role of the evaluation point and the associated sample.

The asymptotic bias and variance for proper and improper estimators are listed in Table~\ref{tbl:asymptotic-properties}. The expressions for the asymptotics of improper estimators agree with the expressions in the original papers after reparametrisation. Because improper BS and LN estimators have a higher-order dependence on the evaluation point than the improper gamma estimator, they have smaller bias near the boundary which has been confirmed by simulations \cite{Jin2003}. The same applies to the improper IG estimator because its bias is cubic in the evaluation point.

The proper estimators with LN, BS, and RIG kernels are new contributions although the LN kernel is equivalent to a log-transform of the samples followed by standard kernel density estimation using Gaussian kernels (see Section~2.11 in \cite{Silverman1986}), which is widely used. The method we have developed is not applicable to proper estimators with RIG kernels because the assumption that $z\left(y,x\right)=\frac{y-x}{\sigma y^{3/2}}$ is monotonic in $y$ is violated.  All estimators achieve the same $n^{-4/5}$ convergence rate as Gaussian KDEs because their MISE is of the form $\sigma^{4}E+\left(\sigma n\right)^{-1}F$, where $E$ and $F$ are numbers that do not depend on $n$ or $\sigma$.

\section{Discussion\label{sec:discussion}}

We have extended balloon and sample-smoothing estimators by a shift parameter and derived their asymptotic properties. Our approach facilitates the unified study of a wide range of density estimators which are subsumed under these two general classes of kernel density estimators. We have demonstrated our method by obtaining the asymptotics of estimators with improper gamma, log-normal, Birnbaum-Saunders, inverse Gaussian and reciprocal inverse Gaussian. Recently, an estimator with proper gamma kernel was proposed \cite{Jeon2014}; we have obtained the first explicit expression for its bias. We have proposed estimators with proper Birnbaum-Saunders and reciprocal inverse Gaussian kernels which have the distinct advantage of yielding normalised density estimates. Plugin expressions for bandwidth estimation are provided for both proper and improper estimators using the log-normal distribution as a reference distribution.

Given the plethora of density estimators, which ones should be used? If the generating density is expected to have a large gradients, the proper gamma estimator and the improper inverse Gaussian (IG) or reciprocal inverse Gaussian (RIG) estimators are suitable because their bias is independent of the first derivative. The improper log-normal, Birnbaum-Saunders, IG and RIG estimators as well as the proper gamma estimator have relatively low bias near the origin and are appropriate for samples that are clustered near zero. Estimators whose effective bandwidth has a higher-order dependence on the evaluation point should be used for heavily-skewed data to avoid spurious wiggles in the tails of the distribution. We believe that proper estimators are preferable because they are guaranteed to yield normalised densities.

In future work, we would like to apply our approach to further kernels and conduct a full simulation study of the finite-sample properties of the estimators.

\appendix

\section{\label{app:proofs}Proofs}
\begin{lem}
\label{lem:integral-balloon}The integral

\[
L=\int\frac{dy}{h\left(x\right)}\phi\left(\frac{y-x-h^{p}\left(x\right)\delta\left(x\right)}{h\left(x\right)}\right)\theta\left(y\right)
\]
is given by
\begin{align*}
L & =\sum_{k=0}^{\infty}A_{k}\left(x\right)\int dz\,\phi\left(z\right)z^{k},\\
\text{where}\quad A_{k}\left(x\right) & =h^{k}\left(x\right)\sum_{j=k}^{\infty}\frac{\theta^{\left(j\right)}\left(x\right)}{k!\left(j-k\right)!}\left[h^{p}\left(x\right)\delta\left(x\right)\right]^{j-k}.
\end{align*}
\end{lem}
\begin{proof}
We make the change of variables 
\[
z=\frac{y-x-h^{p}\left(x\right)\delta\left(x\right)}{h\left(x\right)}
\]
such that, holding $x$ constant,
\begin{align*}
dz & =\frac{dy}{h\left(x\right)}\\
\text{and}\quad L & =\int dz\,\phi\left(z\right)\theta\left(x+h^{p}\left(x\right)\delta\left(x\right)+zh\left(x\right)\right).
\end{align*}
Expanding $\theta$ as a Taylor series about $x$ yields
\[
L=\int dz\,\phi\left(z\right)\sum_{j=0}^{\infty}\frac{\theta^{\left(j\right)}\left(x\right)}{j!}\left(h^{p}\left(x\right)\delta\left(x\right)+zh\left(x\right)\right)^{j}.
\]
Using the binomial theorem and exchanging the order of summation gives
\begin{align*}
L & =\int dz\,\phi\left(z\right)\sum_{j=0}^{\infty}\frac{\theta^{\left(j\right)}\left(x\right)}{j!}\sum_{k=0}^{j}\binom{j}{k}z^{k}h^{k+p\left(j-k\right)}\left(x\right)\delta^{j-k}\left(x\right)\\
 & =\sum_{k=0}^{\infty}A_{k}\left(x\right)\int dz\,\phi\left(z\right)z^{k},\\
\text{where}\quad A_{k}\left(x\right) & =h^{k}\left(x\right)\sum_{j=k}^{\infty}\frac{\theta^{\left(j\right)}\left(x\right)}{k!\left(j-k\right)!}\left[h^{p}\left(x\right)\delta\left(x\right)\right]^{j-k}.
\end{align*}

\end{proof}

\begin{proof}[Proof of Theorem~\ref{thm:balloon-moments}]
From Eqs.~\eqref{eq:generalised-mean} and \eqref{eq:balloon-definition}, the first moment of a shifted balloon estimator is
\[
\left\langle \hat{f}\left(x\right)\right\rangle =\int\frac{dy}{h\left(x\right)}\,K\left(\frac{y-x-h^{p}\left(x\right)\delta\left(x\right)}{h\left(x\right)}\right)f\left(y\right).
\]
Let $\phi\left(z\right)=K\left(z\right)$ and $\theta\left(y\right)=f\left(y\right)$. Then by Lemma~\ref{lem:integral-balloon}
\begin{align}
\left\langle \hat{f}\left(x\right)\right\rangle  & =\sum_{k=0}^{\infty}A_{k}\left(x\right)\int dz\,K\left(z\right)z^{k},\label{eq:balloon-mean-proof}\\
\text{where}\quad A_{k}\left(x\right) & =h^{k}\left(x\right)\sum_{j=k}^{\infty}\frac{f^{\left(j\right)}\left(x\right)}{k!\left(j-k\right)!}\left[h^{p}\left(x\right)\delta\left(x\right)\right]^{j-k}.\nonumber 
\end{align}
Recall that by Definition~\ref{def:order-p-kernel}, the integral in Eq.~\eqref{eq:balloon-mean-proof} vanishes for $0<k<p$ and is equal to unity for $k=0,\,p$ such that
\[
\left\langle \hat{f}\left(x\right)\right\rangle =A_{0}\left(x\right)+A_{p}\left(x\right)+o\left(h^{p}\left(x\right)\right).
\]

According to Eqs.~\eqref{eq:generalised-variance} and \eqref{eq:balloon-definition}, the variance of a shifted balloon estimator is
\begin{equation}
\var\hat{f}\left(x\right)=\frac{1}{n}\int\frac{dy}{h^{2}\left(x\right)}\,K^{2}\left(\frac{y-x-h^{p}\left(x\right)\delta\left(x\right)}{h\left(x\right)}\right)f\left(y\right)-\frac{\left\langle \hat{f}\left(x\right)\right\rangle ^{2}}{n}.\label{eq:balloon-var-proof}
\end{equation}
We use Lemma~\ref{lem:integral-balloon} with $\phi\left(z\right)=K^{2}\left(z\right)$and $\theta\left(y\right)=f\left(y\right)/h\left(x\right)$ such that the variance becomes
\[
\var\hat{f}\left(x\right)=\frac{f\left(x\right)\kappa}{nh\left(x\right)}+o\left(n^{-1}h^{-1}\left(x\right)\right).
\]
The second term in Eq.~\eqref{eq:balloon-var-proof} is absorbed by $o\left(n^{-1}h^{-1}\left(x\right)\right)$ because of its zero-order dependence on $h$$\left(x\right)$.\end{proof}
\begin{lem}
\label{lem:integral-smoothing}The integral 
\[
L=\int\frac{dy}{h\left(y\right)}\phi\left(\frac{y-x+h^{p}\left(y\right)\delta\left(y\right)}{h\left(y\right)}\right)\theta\left(y\right)
\]
is given by
\begin{align*}
L & =\sum_{k=0}^{\infty}B_{k}\left(x\right)\int dz\,\phi\left(z\right)z^{k},\\
\text{where}\quad B_{k}\left(x\right) & =\sum_{j=k}^{\infty}\frac{1}{k!\left(j-k\right)!}\frac{d^{j}}{dx^{j}}\left(\theta\left(x\right)h^{k+p\left(j-k\right)}\left(x\right)\left(-\delta\left(x\right)\right)^{j-k}\right)
\end{align*}
if 
\begin{equation}
z\left(y,x\right)=\frac{y-x+h^{p}\left(y\right)\delta\left(y\right)}{h\left(y\right)}\label{eq:smoothing-change-of-variables}
\end{equation}
is a monotonic function of $y$ for all $x$ such that the inverse $y\left(z,x\right)$ exists.\end{lem}
\begin{proof}
We make the change of variables in Eq.~\eqref{eq:smoothing-change-of-variables} such that, holding $x$ constant, 
\begin{align}
dz & =\left(\frac{\frac{d}{dy}\left(y-x+h^{p}\left(y\right)\delta\left(y\right)\right)}{h\left(y\right)}-\frac{y-x+h^{p}\left(y\right)\delta\left(y\right)}{h\left(y\right)}\frac{h'\left(y\right)}{h\left(y\right)}\right)dy\nonumber \\
 & =\frac{dy}{h\left(y\right)}\frac{d}{dy}\left(y-x+h^{p}\left(y\right)\delta\left(y\right)-zh\left(y\right)\right)\nonumber \\
\text{and}\quad L & =\int dz\,\phi\left(z\right)\times\frac{\theta\left(y\left(z,x\right)\right)}{\left[\frac{d}{dy}\left(y-x+h^{p}\left(y\right)\delta\left(y\right)-zh\left(y\right)\right)\right]_{y=y\left(z,x\right)}}.\label{eq:smoothing-transformed}
\end{align}

For fixed $z$, we express the fraction in the integrand of Eq.~\eqref{eq:smoothing-transformed} as the contour integral 
\begin{equation}
\ell=\frac{1}{2\pi i}\oint_{\gamma}d\tau\,\frac{\theta\left(\tau\right)}{\tau-x+h^{p}\left(\tau\right)\delta\left(\tau\right)-zh\left(\tau\right)},\label{eq:contour-integral}
\end{equation}
where $\tau$ is an integration variable. The contour $\gamma$ encloses the points defined by $\tau=y\left(z,x\right)$ and $\tau=x$. The simple pole at $\tau=y\left(z,x\right)$ has residue (see Theorem 8.15 in \cite{Howie2008}) 
\[
\frac{\theta\left(y\left(z,x\right)\right)}{\left[\frac{d}{dy}\left(y-x+h^{p}\left(y\right)\delta\left(y\right)-zh\left(y\right)\right)\right]_{y=y\left(z,x\right)}}
\]
such that the contour integral in Eq.~\eqref{eq:contour-integral} is indeed equal to the second term in Eq.~\eqref{eq:smoothing-transformed} by Cauchy's residue theorem.

Expanding $\ell$ as a power series about $\tau=x$ yields
\begin{align*}
\ell & =\frac{1}{2\pi i}\oint_{\gamma}d\tau\,\frac{\theta\left(y\right)}{\tau-x}\sum_{j=0}^{\infty}\left(\frac{zh\left(\tau\right)-h^{p}\left(x\right)\delta\left(\tau\right)}{\tau-x}\right)^{j}\\
 & =\frac{1}{2\pi i}\oint_{\gamma}d\tau\,\theta\left(y\right)\sum_{j=0}^{\infty}\sum_{k=0}^{j}\binom{j}{k}\frac{z^{k}h^{k+p\left(j-k\right)}\left(\tau\right)\left(-\delta\left(\tau\right)\right)^{j-k}}{\left(\tau-x\right)^{j+1}}.
\end{align*}
The second equality follows by the binomial theorem. The $j^{\text{th}}$ term gives rise to a pole of order $j+1$ at $\tau=x$ with residue (see Theorem 8.17 in \cite{Howie2008})
\[
R_{j}=\sum_{k=0}^{j}\frac{z^{k}}{k!\left(j-k\right)!}\frac{d^{j}}{dx^{j}}\left(\theta\left(x\right)h^{k+p\left(j-k\right)}\left(x\right)\left(-\delta\left(x\right)\right)^{j-k}\right)
\]
such that $\ell=\sum_{j=0}^{\infty}R_{j}$. Substituting back into Eq.~\eqref{eq:smoothing-transformed} and exchanging the order of summation yields
\begin{align*}
L & =\int dz\,\phi\left(z\right)\sum_{j=0}^{\infty}\sum_{k=0}^{j}\frac{z^{k}}{k!\left(j-k\right)!}\frac{d^{j}}{dx^{j}}\left(\theta\left(x\right)h^{k+p\left(j-k\right)}\left(x\right)\left(-\delta\left(x\right)\right)^{j-k}\right)\\
 & =\sum_{k=0}^{\infty}B_{k}\left(x\right)\int dz\,\phi\left(z\right)z^{k}\\
\text{where}\quad B_{k}\left(x\right) & =\sum_{j=k}^{\infty}\frac{1}{k!\left(j-k\right)!}\frac{d^{j}}{dx^{j}}\left(\theta\left(x\right)h^{k+p\left(j-k\right)}\left(x\right)\left(-\delta\left(x\right)\right)^{j-k}\right).
\end{align*}

\end{proof}

\begin{proof}[Proof of Theorem~\ref{thm:smoothing-moments}]
From Eqs.~\eqref{eq:generalised-mean} and \eqref{eq:smoothing-definition}, the first moment of a shifted sample-smoothing estimator is
\[
\left\langle \hat{f}\left(x\right)\right\rangle =\int\frac{dy}{h\left(y\right)}\,K\left(\frac{y-x+h^{p}\left(y\right)\delta\left(y\right)}{h\left(y\right)}\right)f\left(y\right).
\]
Let $\phi\left(z\right)=K\left(z\right)$ and $\theta\left(y\right)=f\left(y\right)$. Then by Lemma~\ref{lem:integral-smoothing}
\begin{align}
\left\langle \hat{f}\left(x\right)\right\rangle  & =\sum_{k=0}^{\infty}B_{k}\left(x\right)\int dz\,K\left(z\right)z^{k},\label{eq:smoothing-mean-proof}\\
\text{where}\quad B_{k}\left(x\right) & =\sum_{j=k}^{\infty}\frac{1}{k!\left(j-k\right)!}\frac{d^{j}}{dx^{j}}\left(f\left(x\right)h^{k+p\left(j-k\right)}\left(x\right)\left(-\delta\left(x\right)\right)^{j-k}\right).\nonumber 
\end{align}
Recall that by Definition~\ref{def:order-p-kernel}, the integral in Eq.~\eqref{eq:smoothing-mean-proof} vanishes for $0<k<p$ and is equal to unity for $k=0,\,p$ such that
\[
\left\langle \hat{f}\left(x\right)\right\rangle =B_{0}\left(x\right)+B_{p}\left(x\right)+o\left(h^{p}\left(x\right)\right).
\]

According to Eqs.~\eqref{eq:generalised-variance} and \eqref{eq:smoothing-definition}, the variance of a shifted sample-smoothing estimator is
\begin{equation}
\var\hat{f}\left(x\right)=\frac{1}{n}\int\frac{dy}{h^{2}\left(y\right)}\,K^{2}\left(\frac{y-x+h^{p}\left(y\right)\delta\left(y\right)}{h\left(y\right)}\right)f\left(y\right)-\frac{\left\langle \hat{f}\left(x\right)\right\rangle ^{2}}{n}.\label{eq:smoothing-var-proof}
\end{equation}
We use Lemma~\ref{lem:integral-balloon} with $\phi\left(z\right)=K^{2}\left(z\right)$and $\theta\left(y\right)=f\left(y\right)/h\left(y\right)$ such that the variance becomes
\[
\var\hat{f}\left(x\right)=\frac{f\left(x\right)\kappa}{nh\left(x\right)}+o\left(n^{-1}h^{-1}\left(x\right)\right).
\]
The second term in Eq.~\eqref{eq:smoothing-var-proof} is absorbed by $o\left(n^{-1}h^{-1}\left(x\right)\right)$ because of its zero-order dependence on $h$$\left(x\right)$.
\end{proof}

\section{Cross-validation\label{app:cross-validation}}

The quantity 
\begin{equation}
M=\int dx\,\hat{f}^{2}\left(x\right)-\frac{2}{n}\sum_{i=1}^{n}\hat{f}_{i}\left(X_{i}\right)\label{eq:mise-estimator}
\end{equation}
is an unbiased estimator of the MISE of the density estimator $\hat{f}\left(x\right)$, where $\hat{f_{i}}\left(X_{i}\right)$ is the density estimated from all samples except $X_{i}$ evaluated at $X_{i}$. Details can be found in Section~3.4.3 of \cite{Silverman1986}. The second term in Eq.~\eqref{eq:mise-estimator} is straightforward to evaluate whereas the first requires integration of the density estimate over the whole domain. Fortunately, the integral is tractable for the proper gamma estimator in Eq.~\eqref{eq:jeon-kim-estimator}:
\begin{align*}
\int dx\,\hat{f}\left(x\right) & =\frac{1}{n^{2}}\sum_{i,j=1}^{n}\int dx\,\mathrm{G}\left(x;1+\frac{X_{i}}{\sigma^{2}},\sigma^{2}\right)\mathrm{G}\left(x;1+\frac{X_{j}}{\sigma^{2}},\sigma^{2}\right)\\
 & =\frac{1}{n^{2}}\sum_{i,j=1}^{n}\frac{\Gamma\left(1+\frac{X_{i}+X_{j}}{\sigma^{2}}\right)2^{-1-\frac{X_{i}+X_{j}}{\sigma^{2}}}}{\Gamma\left(1+X_{i}/\sigma^{2}\right)\Gamma\left(1+X_{j}/\sigma^{2}\right)}.
\end{align*}

\printbibliography
\end{document}